\documentclass[11pt]{article}
\usepackage[utf8]{inputenc}
\usepackage[english]{babel}
\usepackage{hyphenat}
\usepackage{mathtools}
    \numberwithin{equation}{section}
\usepackage{amsfonts}
\usepackage{amssymb}
\usepackage{amsthm}
\usepackage{stmaryrd}
\usepackage{hyperref}
\usepackage{tikz-cd}
\usepackage{tikz}
    \usetikzlibrary{arrows.meta,positioning,calc}
\usepackage{fullpage}
\usepackage{fancyhdr}

\allowdisplaybreaks

\newtheorem{theorem}{Theorem}[section]
\newtheorem{lemma}[theorem]{Lemma}
\newtheorem{prop}[theorem]{Proposition}
\newtheorem{corollary}[theorem]{Corollary}

\theoremstyle{definition}
\newtheorem{definition}[theorem]{Definition}

\theoremstyle{remark}
\newtheorem*{remark}{Remark}

\theoremstyle{remark}

\newcommand{\Z}{\mathbb{Z}}

\newcommand{\C}{\mathbb{C}}

\newcommand{\I}{\mathrm{i}}

\title{Spin Ruijsenaars--Schneider models are Coulomb branches}
\author{Gleb Arutyunov\footnote{\href{mailto://gleb.arutyunov@desy.de}{\texttt{gleb.arutyunov@desy.de}} (\href{https://orcid.org/0009-0009-8862-6959}{orcid.org/0009-0009-8862-6959}), \href{mailto://lukas.hardi@desy.de}{\texttt{lukas.hardi@desy.de}} (\href{https://orcid.org/0009-0005-0592-8486}{orcid.org/0009-0005-0592-8486}), II. Institut für Theoretische Physik, Universität Hamburg, Luruper Chaussee 149, 22761 Hamburg, Germany} \and Lukas Hardi${}^*$}
\date{}

\begin{document}
\maketitle

\thispagestyle{fancy}
\rhead{hep-th/2603.03048, ZMP-HH/26-3}

\begin{abstract}
\noindent In this paper, we show that the Poisson algebras of homological and $K$-theoretic Coulomb branches of 3d $\mathcal{N}=4$ and 4d $\mathcal{N}=2$ necklace quiver gauge theories provide Poisson structures and Hamiltonians that reproduce the equations of motion of the rational and hyperbolic spin Ruijsenaars--Schneider models, respectively. The construction is carried out in terms of abelianized monopole operators in the GKLO representation, also making the affine Yangian (and, in $K$-theory, quantum toroidal) superintegrability structure manifest. We conjecture that the Poisson algebras of elliptic Coulomb branches similarly reproduce the elliptic spin Ruijsenaars--Schneider model.
\end{abstract}

\newpage

\tableofcontents

\section{Introduction}

Spin Ruijsenaars--Schneider (RS) models are superintegrable models\footnote{Superintegrable models are also known as degenerately integrable models \cite{reshetikhin:2016}.} of $N$ relativistic particles, each carrying $\ell$ magnetically interacting spin degrees of freedom. Their equations of motion were found by Krichever and Zabrodin in \cite{krichever:1995}. Denoting particle positions by $x_i$ and spin degrees of freedom by $a_i^\alpha,c_i^\alpha$, where $i=1,\dots,N$ and $\alpha=1,\dots,\ell$, the equations of motion are given by
\begin{equation}\label{eq:eom}
    \begin{aligned}
        \ddot x_i
        ={}& \sum_{j \neq i} f_{ij} f_{ji} (V(x_i-x_j)-V(x_j-x_i)), \\
        \dot a_i^\alpha
        ={}& -\sum_{j \neq i} (a_i^\alpha-a_j^\alpha) f_{ij} V(x_i-x_j), \\
        \dot c_i^\alpha
        ={}& \sum_{j \neq i} (c_i^\alpha f_{ij} V(x_i-x_j) - c_j^\alpha f_{ji} V(x_j-x_i)),
    \end{aligned}
\end{equation}
where $f_{ij} \coloneq \sum_{\rho=1}^\ell a_i^\rho c_j^\rho$ and $V(z) = \zeta(z)-\zeta(z+\gamma)$ with coupling constant $\gamma$ is the elliptic potential by which particles interact constructed from the Weierstrass zeta function $\zeta(z)$. Rational and hyperbolic degeneration of the potential yields the equations of motion of the rational and hyperbolic spin RS models.

While the equations of motion of spin RS models were found in \cite{krichever:1995}, the underlying Poisson algebra and the Hamiltonian generating the equations of motion where not given. Subsequently, a Poisson algebra reproducing the equations of motion of the rational spin RS model was found in \cite{arutyunov:1998} as the Hamiltonian reduction
\begin{equation}
    (T^* \mathrm{GL}_N \times T^* \C^{N \times \ell}) \sslash_\gamma \mathrm{GL}_N,
\end{equation}
where the coupling constant $\gamma$ is the value of the moment map. The Hamiltonian generating the equations of motion above is simply given by $\operatorname{Tr} g$, where $g$ parametrizes the group element in the base of the cotangent bundle $T^* \mathrm{GL}_N$.

The paper \cite{arutyunov:1998} further exhibited elements $T^{\alpha\beta}(z)$ and $J^{\alpha\beta}[n]$ inside the Poisson algebra that satisfy the Poisson brackets of the classical limit of the Yangian and loop algebra of $\mathfrak{gl}_\ell$. Although the cross relations between $T^{\alpha\beta}(z)$ and $J^{\alpha\beta}[n]$ were not explicitly determined in \cite{arutyunov:1998}, the existence of Yangian and loop algebra generators made it reasonable to conjecture \cite{arutyunov:2025} that the Poisson algebra can be identified with the $N$-truncated affine Yangian of $\mathfrak{gl}_\ell$, which by the work of Braverman--Finkelberg--Nakajima \cite{braverman:2018,nakajima:2017} is identified with the 3d $\mathcal{N}=4$ Coulomb branch of the type $A_{\ell-1}^{(1)}$ quiver where each gauge node has rank $N$ and no flavor nodes are present:
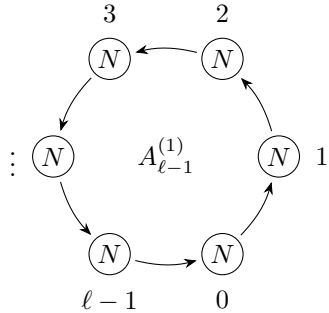
\begin{center}
    \def\l{6}
    \def\radius{1.5cm}
    \begin{tikzpicture}[>=Stealth, every node/.style={font=\small}]
        \foreach \i in {1,...,\l} {
            \pgfmathsetmacro{\angle}{360*(\i-1)/\l}
            \node[circle,draw,inner sep=1.5pt] (v\i) at (\angle:\radius) {$N$};
        }
        \node[right=2pt of v1] (d1) {\small $1$};
        \node[above=2pt of v2] (d2) {\small $2$};
        \node[above=2pt of v3] (d3) {\small $3$};
        \node[left=2pt of v4] (d4) {\small $\vdots$};
        \node[below=2pt of v5] (d5) {\small $\ell-1$};
        \node[below=2pt of v6] (d6) {\small $0$};
        \foreach \i in {1,...,\l} {
            \pgfmathtruncatemacro{\next}{int(mod(\i,\l)+1)}
            \draw[->, shorten >=2pt, shorten <=2pt] (v\i) to[bend right=14] (v\next);
        }
        \node at (0,0) [font=\small,fill=white,inner sep=2pt] {$A_{\ell-1}^{(1)}$};
    \end{tikzpicture}
\end{center}
We will henceforth refer to this quiver as the \emph{necklace quiver} for short.

The purpose of this paper was to use the framework of \cite{bullimore:2015} to show that the equations of motion of the rational spin RS model can be reproduced from a $\gamma$-deformed GKLO representation of the (homological) Coulomb branch of the necklace quiver above. This deformation corresponds to a non-zero mass $\gamma$ of the bifundamental hypermultiplet between the node $\ell-1$ and the node $0$. As a byproduct, we exhibit a series of $L$-operators $L^{\alpha\pm}$ with $\alpha \in \Z/\ell\Z$ living on the Coulomb branch and satisfying the Poisson bracket
\begin{equation}
    \{ L_1^{\alpha\pm}, L_2^{\beta\pm} \} = \pm(\delta^{\alpha\beta} r^\alpha L_1^{\alpha\pm} L_2^{\beta\pm} - \delta^{\alpha\beta} L_1^{\alpha\pm} L_2^{\beta\pm} \underline r^{\alpha+1} + \delta^{\alpha+1,\beta} L_1^{\alpha\pm} \bar r_{21}^\beta L_2^{\beta\pm} - \delta^{\alpha,\beta+1} L_2^{\beta\pm} \bar r^\alpha L_1^{\alpha\pm}),
\end{equation}
where $r^\alpha,\bar r^\alpha,$ and $\underline r^\alpha$ are the matrices found in \cite{arutyunov:1996}. A hierarchy of Poisson-commuting Hamiltonians is then supplied by the traces of powers of the total $L$-operator $L = L^{0+} \cdots L^{\ell-1,+}$. The $L$-operators further assemble into spin variables $a_i^\alpha$ and $c_i^\alpha$ that satisfy the constraint $a_i^1 = 1$, the equations of motion \eqref{eq:eom}, and the Poisson brackets
\begin{align}
        \{ x_i, a_j^\alpha \} ={}& 0, \qquad \{ x_i,c_j^\alpha \} = \delta_{ij} c_j^\alpha, \\[8pt]
        \{ a_i^\alpha,a_j^\beta \} ={}& \frac{1-\delta_{ij}}{x_i-x_j} (a_j^\alpha-a_i^\alpha) (a_i^\beta-a_j^\beta),
        \\[4pt]
        \{ a_i^\alpha,c_j^\beta \} ={}& \frac{1-\delta_{ij}}{x_i-x_j} (a_j^\alpha-a_i^\alpha) c_j^\beta - \delta^{1\beta} L_{ij} a_i^\alpha + \delta^{\alpha\beta} L_{ij},
        \\[4pt]
        \{ c_i^\alpha,c_j^\beta \} ={}& \frac{1-\delta_{ij}}{x_i-x_j} (c_i^\alpha c_j^\beta + c_j^\alpha c_i^\beta) - \delta^{1\alpha} L_{ji} c_j^\beta + \delta^{1\beta} L_{ij} c_i^\alpha.
    \end{align}
with $L_{ij} = -\sum_{\rho=1}^\ell a_i^\rho c_j^\rho/(x_i-x_j+\gamma)$. One can check that the Jacobi identity for these brackets is satisfied, provided one takes the constraint $a_i^1 = 1$ into account. Futhermore, the Poisson brackets of the rescaled spins $\hat a_i^\alpha \coloneq a_i^\alpha/\sum_{\rho=1}^\ell a_i^\rho, \hat c_i^\alpha \coloneq c_i^\alpha \sum_{\rho=1}^\ell a_i^\rho$ coincide with the Poisson brackets from \cite{arutyunov:1998}.

To illustrate the power of the Coulomb branch approach to spin RS models, we further generalize our result for the homological Coulomb branch to the $K$-theoretic Coulomb branch of the same quiver. We again exhibit an $L$-operator algebra of the same form, except that $r^\alpha$, $\bar r^\alpha$, and $\underline r^\alpha$ are now the matrices from \cite{arutyunov:2019b}. The resulting spin variables $a_i^\alpha$ and $c_i^\alpha$ again satisfy the constraint $a_i^1 = 1$, the equations of motion \eqref{eq:eom}, and the Poisson brackets
\begin{align}
    \{ x_i, a_j^\alpha \} ={}& 0, \qquad \{ x_i, c_j^\alpha \} = \delta_{ij} c_j^\alpha, \\[4pt]
    \{ a_i^\alpha,a_j^\beta \} ={}& (1-\delta_{ij}) \bigg( \frac{a_i^\alpha (a_j^\beta-a_i^\beta)}{e^{x_i-x_j}-1} - \frac{a_j^\alpha (a_j^\beta-a_i^\beta)}{1-e^{x_j-x_i}} \bigg) - \delta^{\alpha < \beta} a_j^\alpha a_i^\beta \\
    &+ \tfrac{1}{2} (2\delta^{1=\alpha<\beta} + \delta^{\alpha>\beta>1} + \delta^{1 < \alpha < \beta}) a_i^\alpha a_j^\beta, \nonumber
    \\[4pt]
    \{ a_i^\alpha,c_j^\beta \} ={}& (1-\delta_{ij}) \frac{(a_j^\alpha-a_i^\alpha) c_j^\beta}{1-e^{x_j-x_i}} - \delta^{\alpha\beta} \sum_{\rho=1}^{\alpha-1} a_i^\rho c_j^\rho + \delta^{1\beta} \tilde L_{ij} a_i^\alpha - \delta^{\alpha\beta} \tilde L_{ij} \\
    &+ \tfrac{1}{2}(1 - \delta^{1=\alpha < \beta} + \delta^{\alpha > \beta=1}-\delta^{\alpha\beta}) a_i^\alpha c_j^\beta, \nonumber
    \\[4pt]
    \{ c_i^\alpha,c_j^\beta \} ={}& (1-\delta_{ij}) \bigg( \frac{c_j^\alpha c_i^\beta}{e^{x_i-x_j}-1} + \frac{c_i^\alpha c_j^\beta}{1-e^{x_j-x_i}} \bigg) + \delta^{\alpha < \beta} c_j^\alpha c_i^\beta - \delta^{1\beta} c_i^\alpha \tilde L_{ij} + \delta^{1\alpha} c_j^\beta \tilde L_{ji} \\
    &- \tfrac{1}{2} (2 \delta^{\alpha > \beta=1} + \delta^{\alpha > \beta > 1} + \delta^{1<\alpha<\beta}) c_i^\alpha c_j^\beta, \nonumber
\end{align}
where $\tilde L_{ij} = \sum_{\rho=1}^\ell a_i^\rho c_j^\rho/(e^{x_i-x_j+\gamma}-1)$ and we say that $\delta^{\mathcal{P}}$ is one whenever $\mathcal{P}$ is true and otherwise zero. Finally, we find that the Poisson brackets of the rescaled spins $\hat a_i^\alpha \coloneq a_i^\alpha/\sum_{\rho=1}^\ell a_i^\rho, \hat c_i^\alpha \coloneq c_i^\alpha \sum_{\rho=1}^\ell a_i^\rho$ coincides with the Poisson brackets from \cite{arutyunov:2019,fairon:2026}. The fact that the $K$-theoretic Coulomb branch yields the same Poisson brackets as the multiplicative quiver variety in \cite{fairon:2026} can be seen as a instance of mirror symmetry, by which $K$-theoretic Coulomb branches correspond to multiplicative quiver varieties of the mirror dual quiver.

Poisson structures reproducing the equations of motion of the hyperbolic/trigonometric spin RS models had previously been obtained by way of quasi-Hamiltonian or Poisson reduction \cite{chalykh:2020,arutyunov:2019,fairon:2021} or restricting to a subspace of the phase space \cite{braden:1996}. We expect that the equations of motion of the elliptic spin RS model can also be reproduced systematically from the Poisson algebra of the elliptic Coulomb branch \cite{finkelberg:2020}, which had previously only been achieved in the case $N=2$ \cite{soloviev:2008}.

The rest of the paper is organized as follows:
\begin{itemize}
    \setlength\itemsep{0pt}
    \item Section \ref{section:coulombBranches}: We recall definitions, conventions, and the presentations of the homological and $K$-theoretic Coulomb branch Poisson algebras.
    \item Section \ref{section:rationalRS}: We give the $\gamma$-deformed GKLO realization of the homological Coulomb branch Poisson algebra, exhibit affine Yangian generators, construct one-site $L$-operators associated to each gauge node as well as total $L$-operators, and show how the Coulomb branch Poisson algebra yields a family of Poisson-commuting Hamiltonians that generate the equations of motion of the rational spin RS model. We also give the Poisson algebra of the rational spin variables.
    \item Section \ref{section:hyperbolicRS}: We present the $t$-deformed (multiplicative) GKLO realization of the $K$-theoretic Coulomb branch Poisson algebra, identify the quantum toroidal generators, define one-site and total $L$-operators as well as Poisson-commuting Hamiltonians that generate the hyperbolic spin RS equations. Finally, we give the Poisson relations of the hyperbolic spin variables.
    \item Section \ref{section:conclusion}: Conclusion and outlook.
    \item Section \ref{section:notation}: A small appendix summarizing the relevant notation used in the paper.
\end{itemize}

\section{Homological and $K$-theoretic Coulomb branches}\label{section:coulombBranches}

Let us give the presentation of the abelianized Coulomb branch algebra following \cite{bullimore:2015}:

\begin{definition}\label{def:cohCoulombAlg}
    The \emph{(abelianized) homological Coulomb branch algebra $\mathfrak{C}_{N,\ell}$ of the necklace quiver} is generated as a Poisson algebra by the generators
    \begin{equation}
        q_i^\alpha, \qquad u_i^{\alpha\pm}, \qquad (q_i^\alpha-q_j^\alpha)^{-1},
    \end{equation}
    where the generators are indexed by $i=1,\dots,N$ and $\alpha \in \Z/\ell\Z$, and we adjoin all inverses $(q_i^\alpha-q_j^\alpha)^{-1}$ for which $i \neq j$. These generators are subject to the \emph{homological Euler class relation}
    \begin{equation}
        u_i^{\alpha+} u_i^{\alpha-} = -\frac{\chi^{\alpha+1}(q_i^\alpha) \chi^{\alpha-1}(q_i^\alpha)}{\prod_{j \neq i} (q_i^\alpha-q_j^\alpha)^2},
    \end{equation}
    where we have introduced the \emph{homological gauge polynomial}
    \begin{equation}\label{eq:cohGaugePoly}
        \chi^\alpha(z) \coloneq \prod_{i=1}^N (z-q_i^\alpha),
    \end{equation}
    and subject to the Poisson brackets
    \begin{align}
        \{ q_i^\alpha,q_j^\beta \}
        &= 0,
        \\[5pt]
        \{ q_i^\alpha,u_j^{\beta\pm} \}
        &= \pm\delta_{ij} \delta^{\alpha\beta} u_j^{\beta\pm},
        \\
        \{ u_i^{\alpha\pm},u_j^{\beta\pm} \}
        &= \pm\frac{1-\delta_{ij}\delta^{\alpha\beta}}{q_i^\alpha-q_j^\beta} \kappa^{\alpha\beta} u_i^{\alpha\pm} u_j^{\beta\pm}, \label{eq:monopoleRelation}
        \\
        \{ u_i^{\alpha+},u_j^{\beta-} \}
        &= \delta_{ij} \delta^{\alpha\beta} \frac{\partial}{\partial q_i^\alpha} \frac{\chi^{\alpha+1}(q_i^\alpha) \chi^{\alpha-1}(q_i^\alpha)}{\prod_{j \neq i} (q_i^\alpha-q_j^\alpha)^2},
    \end{align}
    where $\kappa^{\alpha\beta} = 2\delta^{\alpha\beta} - \delta^{\alpha+1,\beta}-\delta^{\alpha,\beta+1}$ is the Cartan matrix of the necklace quiver with $\delta^{\alpha\beta}$ the Kronecker delta on $\Z/\ell\Z$.
\end{definition}

\begin{remark}
    The relation \eqref{eq:monopoleRelation} implies that the right-hand-side is also an element of $\mathfrak{C}_{N,\ell}$, even though it is not generated from the generators as a commutative algebra.
\end{remark}

From the viewpoint of the 3d $\mathcal{N}=4$ quiver gauge theory of the necklace quiver, the diagonal matrices $\operatorname{diag}(q_1^\alpha,\dots,q_N^\alpha)$ have the interpretation of the vacuum expectation value of the scalar field inside the vector multiplet associated to the $\alpha$th gauge node. This vacuum expectation value generically breaks the $U(N)$ gauge group associated to the $\alpha$th gauge node down to $U(1)^{\times N}$ and the $W$-bosons acquire the inverse effective mass $(q_i^\alpha-q_j^\alpha)^{-1}$. Functions in the generators $u_i^{\alpha\pm}$ which are symmetric in the lower index have the interpretation of monopole operators the gauge group $U(N)$ associated to the $\alpha$th gauge node.

Next, we introduce the $K$-theoretic Coulomb branch algebra $\mathfrak{C}_{N,\ell}^K$ following \cite{finkelberg:2018,tsymbaliuk:2023}. We will abuse notation and denote its monopole operators by the same symbols as for the homological Coulomb branch. It should be clear from context whether we are treating the homological or $K$-theoretic case.

\begin{definition}\label{def:KthCoulombAlg}
    The \emph{(abelianized) $K$-theoretic Coulomb branch algebra $\mathfrak{C}_{N,\ell}^K$ of the necklace quiver} is generated as a Poisson algebra by
    \begin{equation}
        (Q_i^\alpha)^{\pm 1/2}, \qquad u_i^{\alpha\pm}, \qquad ((Q_i^\alpha/Q_j^\alpha)^{1/2}-(Q_j^\alpha/Q_i^\alpha)^{1/2})^{-1},
    \end{equation}
    where the generators are indexed by $i=1,\dots,N$ and $\alpha \in \Z/\ell\Z$, and we adjoin all inverses $((Q_i^\alpha/Q_j^\alpha)^{1/2}-(Q_j^\alpha/Q_i^\alpha)^{1/2})^{-1}$ for which $i \neq j$. These generators are subject to the \emph{$K$-theoretic Euler class relation}
    \begin{equation}
        u_i^{\alpha+} u_i^{\alpha-} = -\frac{\chi^{\alpha+1}(Q_i^\alpha) \chi^{\alpha-1}(Q_i^\alpha)}{\prod_{j \neq i} ((Q_i^\alpha/Q_j^\alpha)^{1/2}-(Q_j^\alpha/Q_i^\alpha)^{1/2})^2},
    \end{equation}
    with the \emph{$K$-theoretic gauge polynomial}
    \begin{equation}\label{eq:KthGaugePoly}
        \chi^\alpha(z) \coloneq \prod_{i=1}^N ((z/Q_i^\alpha)^{1/2}-(Q_i^\alpha/z)^{1/2}),
    \end{equation}
    and subject to the Poisson brackets
    \begin{align}
        \{ Q_i^\alpha,Q_j^\beta \}
        &= 0,
        \\[7pt]
        \{ Q_i^\alpha,u_j^{\beta\pm} \}
        &= \pm\delta_{ij} \delta^{\alpha\beta} Q_i^\alpha u_j^{\beta\pm}, \\
        \{ u_i^{\alpha\pm},u_j^{\beta\pm} \}
        &= \pm (1-\delta_{ij}\delta^{\alpha\beta}) \frac{1}{2} \frac{Q_i^\alpha+Q_j^\beta}{Q_i^\alpha-Q_j^\beta} \kappa^{\alpha\beta} u_i^{\alpha\pm} u_j^{\beta\pm}, \\
        \{ u_i^{\alpha+},u_j^{\beta-} \}
        &= \delta_{ij} \delta^{\alpha\beta} Q_i^\alpha \frac{\partial}{\partial Q_i^\alpha} \frac{\chi^{\alpha+1}(Q_i^\alpha) \chi^{\alpha-1}(Q_i^\alpha)}{\prod_{j \neq i} ((Q_i^\alpha/Q_j^\alpha)^{1/2}-(Q_j^\alpha/Q_i^\alpha)^{1/2})^2}.
    \end{align}
    where $\kappa^{\alpha\beta}$ is again the Cartan matrix of the necklace quiver.
\end{definition}

From the viewpoint of gauge theory, the $K$-theoretic Coulomb branch can be seen as the classical limit of the algebra of line operators in 4d $\mathcal{N}=2$ necklace quiver gauge theory compactified on $S^1$ which wrap $S^1$ \cite{cautis:2023}. In this picture, functions in $Q_i^\alpha$ which are symmetric in the lower index correspond to Wilson lines and functions in $u_i^{\alpha\pm}$ which are symmetric in the lower index correspond to 't Hooft lines. A natural Hilbert space representation of quantized $K$-theoretic Coulomb branch algebras is provided by Schur quantization \cite{gaiotto:2024b}, giving a natural Hilbert space quantization of the hyperbolic spin RS model. This quantization appears to differ from the quantum spin RS model defined in \cite{lamers:2022}. A detailed comparison will be the subject of future work.

\section{Rational spin RS models are homological Coulomb branches}\label{section:rationalRS}

\subsection{GKLO representation of the homological Coulomb branch}

We proceed along the lines of \cite{gerasimov:2005} to construct a $\gamma$-deformed GKLO representation of the homological Coulomb branch algebra $\mathfrak{C}_{N,\ell}$. To this end, we give the following

\begin{definition}\label{def:ratDiffOps}
    Let $\mathfrak{A}_{N,\ell}$ be the commutative $\C\llbracket \gamma \rrbracket$-algebra
    \begin{equation}
        \mathfrak{A}_{N,\ell} \coloneq \C\llbracket \gamma \rrbracket[q_i^\alpha,(P_i^\alpha)^{\pm 1}][(q_i^\alpha-q_j^\beta)^{-1}]/J,
    \end{equation}
    where the generators $q_i^\alpha,P_i^\alpha$ have indices $i=1,\dots,N$ and $\alpha \in \Z$, we localize at the elements $q_i^\alpha-q_j^\beta$ for $(i,\alpha) \neq (j,\beta)$, and $J$ is the ideal generated by the cyclic relations
    \begin{equation}
        q_i^{\alpha+\ell} = q_i^\alpha - \gamma, \qquad P_i^{\alpha+\ell} = P_i^\alpha.
    \end{equation}
    We then make $\mathfrak{A}_{N,\ell}$ into a Poisson $\C\llbracket \gamma \rrbracket$-algebra via the log-canonical Poisson bracket
    \begin{equation}
        \{ q_i^\alpha, P_j^\beta \} = \delta_{ij} \delta^{\alpha\beta} P_j^\beta.
    \end{equation}
\end{definition}

\begin{prop}
    There is an injective homomorphism $\psi\colon \mathfrak{C}_{N,\ell} \to \mathfrak{A}_{N,\ell}/\gamma \mathfrak{A}_{N,\ell}$ of Poisson algebras that sends
    \begin{equation}
        \begin{aligned}
            \psi\colon \quad 
            q_i^\alpha \mapsto q_i^\alpha, \qquad
            u_i^{\alpha\pm} \mapsto \pm (P_i^\alpha)^{\pm 1} \frac{\chi^{\alpha \pm 1}(q_i^\alpha)}{\prod_{j \neq i} (q_i^\alpha-q_j^\alpha)} + \gamma \mathfrak{A}_{N,\ell}.
        \end{aligned}
    \end{equation}
\end{prop}

\begin{remark}
    Because of the existence of $\psi$, we justify abusing notation and writing
    \begin{equation}
        u_i^{\alpha\pm} \coloneq \pm (P_i^\alpha)^{\pm 1} \frac{\chi^{\alpha\pm1}(q_i^\alpha)}{\prod_{j \neq i} (q_i^\alpha-q_j^\alpha)} \in \mathfrak{A}_{N,\ell}
    \end{equation}
    as a shorthand. Henceforth we will only be working with these \emph{$\gamma$-deformed} monopole operators inside $\mathfrak{A}_{N,\ell}$.
\end{remark}

\begin{proof}
    When $\ell > 2$, we compute the Poisson brackets inside $\mathfrak{A}_{N,\ell}$ to be
    \begin{equation*}
        \{ u_i^{\alpha\pm},u_j^{\beta\pm} \} = \pm(1-\delta_{ij}\delta^{\alpha\beta})\kappa^{\alpha\beta} u_i^{\alpha\pm} u_j^{\beta\pm}
        \begin{cases}
            \frac{1}{q_i^\ell-q_j^{\ell-1}}, & (\alpha,\beta) = (0,\ell-1), \\
            \frac{1}{q_i^{\ell-1}-q_j^\ell}, & (\alpha,\beta) = (\ell-1,0), \\
            \frac{1}{q_i^\alpha-q_j^\beta}, & \text{otherwise}.
        \end{cases}
    \end{equation*}
    For $\ell=2$, we have
    \begin{equation*}
        \{ u_i^{\alpha\pm},u_j^{\beta\pm} \} = \pm(1-\delta_{ij}\delta^{\alpha\beta})\kappa^{\alpha\beta} u_i^{\alpha\pm} u_j^{\beta\pm}
        \begin{cases}
            \frac{1}{2} \big( \frac{1}{q_i^0-q_j^1} + \frac{1}{q_i^2-q_j^1} \big), & (\alpha,\beta) = (0,1), \\
            \frac{1}{2} \big( \frac{1}{q_i^1-q_j^0} + \frac{1}{q_i^1-q_j^2} \big), & (\alpha,\beta) = (1,0), \\
            \frac{1}{q_i^\alpha-q_j^\beta}, & \text{otherwise},
        \end{cases}
    \end{equation*}
    and for $\ell=1$, we have
    \begin{equation*}
        \{ u_i^{0\pm}, u_j^{0\pm}  \} = \pm (1-\delta_{ij})u_i^{0\pm} u_j^{0\pm} \bigg( \frac{2}{q_i^0-q_j^0} - \frac{1}{q_i^0-q_j^1} - \frac{1}{q_i^1-q_j^0} \bigg).
    \end{equation*}
    Since $\frac{1}{q_i^\ell - q_j^{\ell-1}} \equiv \frac{1}{q_i^0 - q_j^{\ell-1}} \mod \gamma \mathfrak{A}_{N,\ell}$, the result follows.
\end{proof}

\subsection{Affine Yangian of $\mathfrak{gl}_\ell$}

With this in hand, we may define the generating series
\begin{equation}
    e^\alpha(z) \coloneq \sum_{i=1}^N \frac{u_i^{\alpha+}}{z-q_i^\alpha} \in \mathfrak{A}_{N,\ell}\llbracket z^{-1} \rrbracket, \qquad f^\alpha(z) \coloneq \sum_{i=1}^N \frac{u_i^{\alpha-}}{z-q_i^\alpha} \in \mathfrak{A}_{N,\ell}\llbracket z^{-1} \rrbracket,
\end{equation}
as well as
\begin{equation}
    h^\alpha(z) \coloneq \frac{\chi^{\alpha+1}(z) \chi^{\alpha-1}(z)}{\chi^\alpha(z)^2} \in \mathfrak{A}_{N,\ell}\llbracket z^{-1} \rrbracket,
\end{equation}
and check that they satisfy the relations of the classical limit of the affine Yangian:

\begin{prop}
    The generating series $e^\alpha(z),f^\alpha(z),\chi^\alpha(z)$ define a representation of the classical limit of the $N$-truncated affine Yangian of $\mathfrak{gl}_\ell$ in the sense that $\chi^\alpha(z)$ is a polynomial of degree $N$ and the relations
    \begin{align}
        \{ \chi^\alpha(z), \chi^\beta(w) \} &=0, \\
        \{ e^\alpha(z),f^\beta(w) \}
        &= -\delta^{\alpha\beta} \frac{h^\alpha(z)-h^\beta(w)}{z-w}, \\
        \{ \chi^\alpha(z),e^\beta(w) \}
        &= \delta^{\alpha\beta} \chi^\alpha(z) \frac{e^\beta(z)-e^\beta(w)}{z-w}, \\
        \{ \chi^\alpha(z),f^\beta(w) \}
        &= -\delta^{\alpha\beta} \chi^\alpha(z) \frac{f^\beta(z)-f^\beta(w)}{z-w}
    \end{align}
    are satisfied.
\end{prop}

\begin{remark}
    This representation is nothing but the classical limit of the GKLO representation of the affine Yangian of $\mathfrak{gl}_\ell$ \cite{gerasimov:2005}. We note that $e^\alpha(z),f^\alpha(z)$ for $\alpha = 1,\dots,\ell-1$ and $\chi^\alpha(z)$ for $\alpha = 0,\dots,\ell-1$ generate the (finite) Yangian of $\mathfrak{gl}_\ell$ with quantum determinant
    \begin{equation}
        \operatorname{qdet}(z) = \prod_{\alpha=0}^{\ell-1} \frac{\chi^{\alpha+1}(z)}{\chi^\alpha(z)} = \frac{\chi^0(z+\gamma)}{\chi^0(z)}.
    \end{equation}
    When $\gamma = 0$, it follows that $\operatorname{qdet}(z) = 1$, which reduces us to the Yangian of $\mathfrak{sl}_\ell$. In that sense, $\gamma$ is the charge under the center of $\mathfrak{gl}_\ell$.
\end{remark}

\begin{corollary}
    The zero modes
    \begin{equation}
        E^\alpha \coloneq \frac{1}{2\pi\I} \oint_\infty e^\alpha(z) dz = \sum_{i=1}^N u_i^{\alpha+}, \qquad F^\alpha \coloneq \frac{1}{2\pi\I} \oint_\infty f^\alpha(z) dz = \sum_{i=1}^N u_i^{\alpha-},
    \end{equation}
    and
    \begin{equation}
        H^\alpha \coloneq \frac{1}{2\pi\I} \oint_\infty h^\alpha(z) dz = \sum_{i=1}^N (2q_i^\alpha - q_i^{\alpha+1} - q_i^{\alpha-1}),
    \end{equation}
    define a representation of $\widehat{\mathfrak{sl}}_\ell$ in the sense that
    \begin{align}
        \{ H^\alpha,H^\beta \} &= 0,
        \\[4pt]
        \{ H^\alpha,E^\beta \} &= \kappa^{\alpha\beta} E^\beta,
        \\[4pt]
        \{ H^\alpha,F^\beta \} &= -\kappa^{\alpha\beta} F^\beta,
        \\[4pt]
        \{ E^\alpha,F^\beta \} &= \delta^{\alpha\beta} H^\alpha.
    \end{align}
\end{corollary}

\subsection{$L$-operator algebra}

Our goal is to exhibit $\mathfrak{A}_{N,\ell}$ as the Poisson algebra of the rational spin RS model. To make such a connection, it is useful to have an $L$-operator algebra at our disposal.

\begin{definition}
    Introduce the one-site $L$-operators
    \begin{equation}
        L_{ij}^{\alpha\pm} \coloneq \frac{u_j^{\alpha+1,\pm}}{q_j^{\alpha+1}-q_i^\alpha} \in \mathfrak{A}_{N,\ell},
    \end{equation}
    as well as the total $L$-operator
    \begin{equation}\label{eq:totLRat}
        L \coloneq L^{0+} \cdots L^{\ell-1,+}.
    \end{equation}
\end{definition}

\begin{remark}
    We note that $L^{\alpha+\ell,\pm} = L^{\alpha\pm}$ by the cyclic relations of $\mathfrak{A}_{N,\ell}$.
\end{remark}

\begin{prop}
    The one-site $L$-operators satisfy the Poisson brackets
    \begin{equation}
        \{ L_1^{\alpha\pm}, L_2^{\beta\pm} \} = \pm(\delta^{\alpha\beta} r^\alpha L_1^{\alpha\pm} L_2^{\beta\pm} - \delta^{\alpha\beta} L_1^{\alpha\pm} L_2^{\beta\pm} \underline r^{\alpha+1} + \delta^{\alpha+1,\beta} L_1^{\alpha\pm} \bar r_{21}^\beta L_2^{\beta\pm} - \delta^{\alpha,\beta+1} L_2^{\beta\pm} \bar r^\alpha L_1^{\alpha\pm}),
    \end{equation}
    where we have used the matrices from \cite{arutyunov:1996}:
    \begin{align}\label{eq:ratMat}
        r^\alpha &\coloneq \sum_{i \neq j} \frac{1}{q_i^\alpha-q_j^\alpha} (e_{ii}-e_{ij}) \otimes (e_{jj} - e_{ji}), \\
        \bar r^\alpha &\coloneq \sum_{i \neq j} \frac{1}{q_i^\alpha-q_j^\alpha} (e_{ii}-e_{ij}) \otimes e_{jj},
        \\[4pt]
        \underline r^\alpha &\coloneq \sum_{i \neq j} \frac{1}{q_i^\alpha-q_j^\alpha} (e_{ij} \otimes e_{ji} - e_{ii} \otimes e_{jj}).
    \end{align}    
\end{prop}

\begin{corollary}
    The total $L$-operator satisfies
    \begin{equation}
        \{ L_1, L_2 \} = r^0 L_1 L_2 - L_1 L_2 \underline r^0 + L_1 \bar r_{21}^0 L_2 - L_2 \bar r^0 L_1,
    \end{equation}
    which reproduces the Poisson bracket of the Lax matrix from \cite{arutyunov:1998}.
\end{corollary}

\begin{proof}
    This follows from the Poisson algbera of the one-site $L$-operators and the identity
    \begin{equation*}
        r^\alpha + \bar r_{21}^\alpha - \bar r^\alpha - \underline r^\alpha = 0.
        \vspace{-1.5\baselineskip}
    \end{equation*}
\end{proof}

\begin{corollary}
    The Hamiltonians $H[n] \coloneq \operatorname{Tr} L^n$ are mutually Poisson-commuting.
\end{corollary}

\begin{prop}
    The Hamiltonians $H[n]$ are central with respect to $\widehat{\mathfrak{sl}}_\ell$:
    \begin{equation}
        \{ H[n], E^\alpha \} = 0, \qquad \{ H[n], F^\alpha \} = 0, \qquad \{ H[n], H^\alpha \} = 0.
    \end{equation}
\end{prop}

\begin{proof}
    Let us consider the bracket $\{ H[n],E^\alpha \}$ as an example. From the $L$-operator algebra, we derive the bracket
    \begin{equation*}
        \begin{aligned}
            \{ L_1^{\alpha+},u_2^{\beta+1,+} \}
            ={}& {-\delta^{\alpha\beta}} L_1^{\alpha+} u_2^{\beta+1,+} \underline r^{\alpha+1} - \delta^{\alpha,\beta+1} u_2^{\beta+1,+} \bar r^\alpha L_1^{\alpha+} - \delta^{\alpha\beta} Z L_1^{\alpha+} L_2^{\beta+} + \delta^{\alpha+1,\beta} L_1^{\alpha+} Z L_2^{\beta+}.
        \end{aligned}
    \end{equation*}
    with $Z = \sum_{i=1}^N e_{ii} \otimes e_i^t$. Then
    \begin{equation*}
        \begin{aligned}
            \{ H[n],E^\alpha \}
            ={}& \sum_{\mu=0}^{n\ell-1} \operatorname{Tr}_1 L_1^{0,+} \cdots L_1^{\mu-1,+} \{ L_1^\mu, u_2^{\alpha,+} \} L_1^{\mu+1,+} \cdots L_1^{n\ell-1,+} e_2 \\
            ={}& -\sum_{\mu=1}^{n\ell-1} \delta^{\mu\alpha} \operatorname{Tr}_1 L_1^{0,+} \cdots L_1^{\alpha-1,+} u_2^{\alpha,+} (\underline r^{\alpha} + \bar r^\alpha) L_1^{\alpha,+} \cdots L_1^{n\ell-1,+} e_2 - \delta^{\ell\alpha} \operatorname{Tr}_1 L_1^n u_2^{\alpha,+} (\underline r^{0}+\bar r^0) e_2 \\
            & + \delta^{0,\alpha-1} \operatorname{Tr}_1 (L_1^n Z L_2^{\alpha-1,+}-Z  L_2^{\alpha-1,+} L_1^n) e_2 \\
            ={}& 0,
        \end{aligned}
    \end{equation*}
    where we have used $(\underline r^\alpha + \bar r^\alpha) e_2 = 0$.
\end{proof}

\subsection{Superintegrability}

We have already seen that the Hamiltonians $H[n]$ Poisson-commute with the generators $E^\alpha,F^\alpha,H^\alpha$, which define a representation of the loop algebra $\widehat{\mathfrak{sl}}_\ell$. It turns out that this representation factors through a representation of the loop algebra $L(\mathfrak{gl}_\ell)$ with generators $J^{\alpha\beta}[n] \in \mathfrak{A}_{N,\ell}$, which can be expressed in terms of the $L$-operators. To see this, let
\begin{equation}
    V_{\alpha,\beta}^\pm \coloneq u^{\alpha\pm} L^{\alpha\pm} \cdots L^{\beta-1,\pm} e
\end{equation}
with $u^{\alpha\pm} = (u_1^{\alpha\pm},\dots,u_N^{\alpha\pm})$ and $e = (1,\dots,1)^t$. Then we can define
\begin{equation}
    J^{\alpha\beta}[0] \coloneq
    \begin{cases}
        V_{\beta,\alpha-1}^-, & \alpha > \beta \\
        \sum_{i=1}^N (q_i^\alpha - q_i^{\alpha-1}), & \alpha = \beta \\
        V_{\alpha,\beta-1}^+, & \alpha < \beta
    \end{cases},
\end{equation}
as well as
\begin{equation}
    J^{\alpha\beta}[-n] \coloneq V_{\beta,\alpha+n\ell-1}^-, \qquad J^{\alpha\beta}[n] \coloneq V_{\alpha,\beta+n\ell-1}^+.
\end{equation}
for $n > 0$. The coefficients $J^{\alpha\beta}[n]$ assemble into an $\ell \times \ell$ matrix $J[n]$.

\begin{prop}
    The generators $J^{\alpha\beta}[n]$ define a representation of the loop algebra $L(\mathfrak{gl}_\ell)$:
    \begin{equation}
        \{ J^{\alpha\beta}[n],J^{\mu\nu}[m] \} = \delta^{\mu\beta} J^{\alpha\nu}[n+m] - \delta^{\alpha\nu} J^{\mu\beta}[n+m].
\end{equation}
\end{prop}

\begin{remark}
    It was already clear from \cite[\S 6.6]{bullimore:2015} that $J^{\alpha\beta}[0]$ satisfies the relations of $\mathfrak{gl}_\ell$.
\end{remark}

\begin{lemma}
    We have $\operatorname{tr} J[n] = -\gamma H[n]$, where $\operatorname{tr}$ is the operation of taking the trace of an $\ell \times \ell$ matrix.
\end{lemma}

\begin{remark}
    In particular, the generators $J^{\alpha\beta}[n]$ define a representation of $L(\mathfrak{sl}_\ell)$ when $\gamma = 0$, which again exhibits $\gamma$ as the charge under the center of $\mathfrak{gl}_\ell$.
\end{remark}

\begin{proof}
    This follows from a telescopic argument similar to lemma \ref{lemma:ratLOp}.
\end{proof}

From this, we see that the Hamiltonians $H[n]$ are in the center of the subalgebra of $\mathfrak{A}_{N,\ell}$ generated by $J^{\alpha\beta}[n]$ and are thus part of a larger Poisson-commutative subalgebra given by the $N\ell$ algebraically independent Hamiltonians
\begin{equation}
    J^{\alpha\alpha}[n], \qquad n=1,\dots,N, \quad \alpha=1,\dots,\ell.
\end{equation}
Since $\mathfrak{A}_{N,\ell}$ has $2N\ell$ algebraically independent generators as a $\C\llbracket \gamma \rrbracket$-algebra, we conclude that $\mathfrak{A}_{N,\ell}$ describes an integrable model in the sense of the Liouville theorem. In fact, $\mathfrak{A}_{N,\ell}$ describes a \emph{superintegrable} model, since the central Hamiltonians $H[n]$ do not just commute among each other, but also commute with the loop algebra $L(\mathfrak{gl}_\ell)$.

\subsection{Equations of motion}

In this section, we show that the superintegrable model defined by the Poisson algebra $\mathfrak{A}_{N,\ell}$ and the Hamiltonians $J^{\alpha\alpha}[n]$ is the spin Ruijsenaars--Schneider model introduced by Krichever and Zabrodin \cite{krichever:1995}. To show the identification, we introduce
\begin{equation}
    a^\alpha \coloneq L^{0+} \cdots L^{\alpha-2,+} e, \qquad c^\alpha \coloneq u^{\alpha+} L^{\alpha+} \cdots L^{\ell-1,+},
\end{equation}
for $\alpha=1,\dots,\ell$. Notice that $a^1 = e$, in other words, we have the constraints $a_i^1 = 1$ for $i=1,\dots,N$. This should be contrasted with \cite{arutyunov:1998,chalykh:2020,arutyunov:2019,fairon:2026}, where the spin vectors satisfy the alternative constraint $\sum_\rho a_i^\rho = 1$. However, they differ only by an overall rescaling of the spin variables.

\begin{lemma}\label{lemma:ratLOp}
    The total $L$-operator can be written as
    \begin{equation}
        L_{ij} = -\frac{\sum_{\rho=1}^\ell a_i^\rho c_j^\rho}{q_i^0-q_j^\ell}.
\end{equation}
\end{lemma}

\begin{proof}
Indeed,
\begin{equation*}
    \begin{aligned}
        \sum_{\rho=1}^\ell &a_i^\rho c_j^\rho
        = \sum_{\rho=1}^\ell \sum_{k,l=1}^N (L^0 \cdots L^{\rho-2})_{il} u_k^{\rho+} (L^\rho \cdots L^{\ell-1})_{kj} \\
        ={}& \sum_{\rho=1}^\ell \sum_{k,l=1}^N (L^0 \cdots L^{\rho-2})_{il} L_{lk}^{\rho-1} (q_k^\rho-q_l^{\rho-1}) (L^\rho \cdots L^{\ell-1})_{kj} \\
        ={}& \sum_{\rho=1}^\ell \sum_{k=1}^N (L^0 \cdots L^{\rho-1})_{ik} (L^\rho \cdots L^{\ell-1})_{kj} q_k^\rho - \sum_{\rho=1}^\ell \sum_{l=1}^N (L^0 \cdots L^{\rho-2})_{il} (L^{\rho-1} \cdots L^{\ell-1})_{lj} q_l^{\rho-1} \\
        ={}& L_{ij} (q_j^\ell-q_i^0),
    \end{aligned}
\end{equation*}
which yields the result.
\end{proof}

\begin{theorem}
    The time evolution under the Hamiltonian $H \coloneq \gamma H[1]$ reproduces the equations of motion \eqref{eq:eom} with the identification $x_i = q_i^0$ and the rational potential $V(z) = \frac{1}{z}-\frac{1}{z+\gamma}$.
\end{theorem}

\begin{proof}
    Using the $L$-operator algebra, we find
    \begin{align*}
        \dot x_i &= \{ H,q_i^0 \} = -\gamma L_{ii},
        \\[7pt]
        \ddot x_i
        &= \{ H, \{ H,q_i^0 \} \} = \gamma^2 \sum_{j(\neq i)} \frac{2}{q_i^0-q_j^0} L_{ij} L_{ji}, \\
        \dot a_i^\alpha
        &= \{ H,a_i^\alpha \} = \gamma \sum_{j(\neq i)} \frac{1}{q_i^0-q_j^0} (a_i^\alpha-a_j^\alpha) L_{ij}, \\
        \dot c_i^\alpha
        &= \{ H,c_i^\alpha \} = -\gamma \sum_{j(\neq i)} \frac{1}{q_i^0-q_j^0} (c_i^\alpha L_{ij} + c_j^\alpha L_{ji}).
    \end{align*}
    We then use lemma \ref{lemma:ratLOp} and the identities
    \begin{align*}
        \gamma \frac{1}{q_i^0-q_j^0} \frac{1}{q_i^0 - q_j^\ell} &= V(x_i-x_j), \\
        \gamma^2 \frac{2}{q_i^0-q_j^0} \frac{1}{q_i^0 - q_j^\ell} \frac{1}{q_j^0 - q_i^\ell} &= V(x_i-x_j)-V(x_j-x_i)
    \end{align*}
    to arrive at the desired equations of motion.
\end{proof}

\subsection{Poisson bracket of the spin variables}

\begin{prop}
    The Poisson brackets of the spins can be written as
    \begin{align}
        \{ q_i^0, a_j^\alpha \} ={}& 0, \qquad \{ q_i^0,c_j^\alpha \} = \delta_{ij} c_j^\alpha, \\[8pt]
        \{ a_i^\alpha,a_j^\beta \} ={}& \frac{1-\delta_{ij}}{q_i^0-q_j^0} (a_j^\alpha-a_i^\alpha) (a_i^\beta-a_j^\beta),
        \\[4pt]
        \{ a_i^\alpha,c_j^\beta \} ={}& \frac{1-\delta_{ij}}{q_i^0-q_j^0} (a_j^\alpha-a_i^\alpha) c_j^\beta - \delta^{1\beta} L_{ij} a_i^\alpha + \delta^{\alpha\beta} L_{ij},
        \\[4pt]
        \{ c_i^\alpha,c_j^\beta \} ={}& \frac{1-\delta_{ij}}{q_i^0-q_j^0} (c_i^\alpha c_j^\beta + c_j^\alpha c_i^\beta) - \delta^{1\alpha} L_{ji} c_j^\beta + \delta^{1\beta} L_{ij} c_i^\alpha.
    \end{align}
\end{prop}

\begin{proof}
    Let us consider the Poisson bracket $\{ a_1^\alpha,a_2^\beta \}$ as an example, where $1$ and $2$ label auxiliary spaces. Introduce the partial monodromies $L^{\alpha,\beta} \coloneq L^\alpha \cdots L^\beta$. Then
    \begin{equation*}
        \begin{aligned}
            \{ &a_1^\alpha,a_2^\beta \} = \sum_{\mu=0}^{\alpha-2} \sum_{\nu=0}^{\beta-2} L_1^{0,\mu-1} L_2^{0,\nu-1} \{ L_1^\mu, L_2^\nu \} L_1^{\mu+1,\alpha-2} L_2^{\nu+1,\beta-2} e_1 e_2 \\
            ={}& \sum_{\mu=0}^{\alpha-2} \sum_{\nu=0}^{\beta-2} \delta^{\mu\nu} L_1^{0,\mu-1} L_2^{0,\nu-1} r^\mu L_1^{\mu,\alpha-2} L_2^{\nu,\beta-2} e_1 e_2 - \sum_{\mu=1}^{\alpha-1} \sum_{\nu=1}^{\beta-1} \delta^{\mu\nu} L_1^{0,\mu-1} L_2^{0,\nu-1} \underline r^\mu L_1^{\mu,\alpha-2} L_2^{\nu,\beta-2} e_1 e_2 \\
            &+ \sum_{\mu=1}^{\alpha-1} \sum_{\nu=1}^{\beta-2} \delta^{\mu\nu} L_1^{0,\mu-1} L_2^{0,\nu-1} \bar r_{21}^\mu L_1^{\mu,\alpha-2} L_2^{\nu,\beta-2} e_1 e_2 - \sum_{\mu=1}^{\alpha-2} \sum_{\nu=1}^{\beta-1} \delta^{\mu\nu} L_1^{0,\mu-1} L_2^{0,\nu-1} \bar r^\mu L_1^{\mu,\alpha-2} L_2^{\nu,\beta-2} e_1 e_2 \\
            ={}& \delta_{\alpha>1,\beta>1} r^0 a_1^\alpha a_2^\beta = r^0 a_1^\alpha a_2^\beta,
        \end{aligned}
    \end{equation*}
    where we used the identities
    \begin{equation*}
        \underline r^\alpha e_1 e_2 = 0, \qquad (\underline r^\alpha + \bar r^\alpha) e_2 = 0, \qquad (\underline r^\alpha - \bar r_{21}^\alpha) e_1 = 0, \qquad r e_1 = 0, \qquad r e_2 = 0.
    \end{equation*}
    The other brackets follow similarly. In total, we arrive at
    \begin{align*}
        \{ a_1^\alpha,a_2^\beta \} ={}& r^0 a_1^\alpha a_2^\beta \\
        \{ a_1^\alpha,c_2^\beta \} ={}& {-c_2^\beta} \bar r^0 a_1^\alpha - \delta^{1\beta} Z a_1^\alpha L_2 + \delta^{\alpha\beta} Z e_1 L_2 \\
        \{ c_1^\alpha,c_2^\beta \} ={}& {-c_1^\alpha} c_2^\beta \underline r^0 + \delta^{1\beta} c_1^\alpha Z L_2 - \delta^{1\alpha} c_2^\beta Z_{21} L_1.
    \end{align*}
    where $Z = \sum_{i=1}^N e_{ii} \otimes e_i^t$. In components, these are the claimed Poisson brackets.
\end{proof}

\begin{remark}
    Due to the constraint $a_i^1 = 1$, the ``physical'' spin degrees of freedom of the rational spin RS model are actually restricted to $a^\alpha,c^\alpha$ for $\alpha > 1$. We note that they form a quadratic Poisson algebra.
\end{remark}

\begin{corollary}
The rescaling $\hat a_i^\alpha \coloneq a_i^\alpha/\sum_\rho a_i^\rho, \hat c_i^\alpha \coloneq c_i^\alpha \sum_\rho a_i^\rho, \hat L_{ij} \coloneq \frac{\sum_\rho \hat a_i^\rho \hat c_j^\rho}{q_i^0-q_j^0+\gamma}$ reproduces the Poisson brackets from \cite{arutyunov:1998}:
\begin{align}
    \{ q_i^0, \hat a_j^\alpha \} &= 0, \qquad \{ q_i^0, \hat c_j^\alpha \} = \delta_{ij} \hat c_j^\alpha, \\[8pt]
    \{ \hat a_i^\alpha, \hat a_j^\beta \} &= \frac{1-\delta_{ij}}{q_i^0-q_j^0} (\hat a_j^\alpha-\hat a_i^\alpha) (\hat a_i^\beta-\hat a_j^\beta),
    \\[4pt]
    \{ \hat a_i^\alpha, \hat c_j^\beta \} &= \frac{1-\delta_{ij}}{q_i^0-q_j^0} (\hat a_j^\alpha-\hat a_i^\alpha) \hat c_j^\beta + \hat a_i^\alpha \hat L_{ij} - \delta^{\alpha\beta} \hat L_{ij},
    \\[4pt]
    \{ \hat c_i^\alpha, \hat c_j^\beta \} &= \frac{1-\delta_{ij}}{q_i^0-q_j^0} (\hat c_i^\alpha \hat c_j^\beta + \hat c_j^\alpha \hat c_i^\beta) - \hat c_i^\alpha \hat L_{ij} + \hat c_j^\beta \hat L_{ji}.
\end{align}
\end{corollary}

\section{Hyperbolic spin RS models are $K$-theoretic Coulomb branches}\label{section:hyperbolicRS}

\subsection{GKLO representation of the $K$-theoretic Coulomb branch}

For the hyperbolic case, we proceed in essentially the same way, starting with a $t$-deformation of the GKLO presentation of the $K$-theoretic Coulomb branch written down in \cite{finkelberg:2019,tsymbaliuk:2023}. We will use essentially the same notation as in the rational case for this section, as no confusion should arise.

\begin{definition}\label{def:hypDiffOps}
    Let $\mathfrak{A}_{N,\ell}^K$ be the commutative $\C\llbracket t-1 \rrbracket$-algebra
    \begin{equation}
        \mathfrak{A}_{N,\ell}^K \coloneq \C\llbracket t-1 \rrbracket[(Q_i^\alpha)^{\pm 1/2},(P_i^\alpha)^{\pm 1}][((Q_i^\alpha/Q_j^\beta)^{1/2}-(Q_j^\beta/Q_i^\alpha)^{1/2})^{-1}]/J,
    \end{equation}
    where the generators $Q_i^\alpha,P_i^\alpha$ have indices $i=1,\dots,N$ and $\alpha \in \Z$, we localize at the elements $(Q_i^\alpha/Q_j^\beta)^{1/2}-(Q_j^\beta/Q_i^\alpha)^{1/2}$ for $(i,\alpha) \neq (j,\beta)$, and $J$ is the ideal generated by the cyclic relations
    \begin{equation}
        Q_i^{\alpha+\ell} = t Q_i^\alpha, \qquad P_i^{\alpha+\ell} = P_i^\alpha.
    \end{equation}
    We then make $\mathfrak{A}_{N,\ell}^K$ into a Poisson $\C\llbracket t-1 \rrbracket$-algebra via the doubly log-canonical Poisson bracket
    \begin{equation}
        \{ Q_i^\alpha,P_j^\beta \} = \delta_{ij} \delta^{\alpha\beta} Q_i^\alpha P_j^\beta.
    \end{equation}
\end{definition}

\begin{prop}
    There exists an injective homomorphism $\psi^K\colon \mathfrak{C}_{N,\ell}^K \to \mathfrak{A}_{N,\ell}^K/(t-1)\mathfrak{A}_{N,\ell}^K$ of Poisson algebras that sends
    \begin{equation}
        \psi^K\colon \quad Q_i^\alpha \mapsto Q_i^\alpha, \qquad u_i^{\alpha\pm} \mapsto \pm (P_i^\alpha)^{\pm 1} \frac{\chi^{\alpha \pm 1}(Q_i^\alpha)}{\prod_{j \neq i} ((Q_i^\alpha/Q_j^\alpha)^{1/2}-(Q_j^\alpha/Q_i^\alpha)^{1/2})} + (t-1) \mathfrak{A}_{N,\ell}^K.
    \end{equation}
\end{prop}

\begin{remark}
    We again use the existence of $\psi^K$ as a justification for writing
    \begin{equation}
        u_i^{\alpha\pm} \coloneq \pm (P_i^\alpha)^{\pm 1} \frac{\chi^{\alpha \pm 1}(Q_i^\alpha)}{\prod_{j \neq i} ((Q_i^\alpha/Q_j^\alpha)^{1/2}-(Q_j^\alpha/Q_i^\alpha)^{1/2})} \in \mathfrak{A}_{N,\ell}^K
    \end{equation}
    as a shorthand.
\end{remark}

\begin{proof}
    We compute the Poisson brackets inside $\mathfrak{A}_{N,\ell}^K$ for $\ell > 2$ to be
    \begin{equation*}
        \{ u_i^{\alpha\pm},u_j^{\beta\pm} \} = \pm(1-\delta_{ij}\delta^{\alpha\beta}) \kappa^{\alpha\beta} u_i^{\alpha\pm} u_j^{\beta\pm}
        \begin{cases}
            \frac{1}{2} \frac{Q_i^\ell+Q_j^{\ell-1}}{Q_i^\ell-Q_j^{\ell-1}}, & (\alpha,\beta) = (0,\ell-1) \\
            \frac{1}{2} \frac{Q_i^{\ell-1}+Q_j^\ell}{Q_i^{\ell-1}-Q_j^\ell}, & (\alpha,\beta) = (\ell-1,0) \\
            \frac{1}{2} \frac{Q_i^\alpha+Q_j^\beta}{Q_i^\alpha-Q_j^\beta}, & \text{otherwise},
        \end{cases}
    \end{equation*}
    while for $\ell=2$, we get
    \begin{equation*}
        \{ u_i^{\alpha\pm},u_j^{\beta\pm} \} = \pm(1-\delta_{ij}\delta^{\alpha\beta})\kappa^{\alpha\beta} u_i^{\alpha\pm} u_j^{\beta\pm}
        \begin{cases}
            \frac{1}{4} \frac{Q_i^0+Q_j^{\ell-1}}{Q_i^0-Q_j^{\ell-1}} + \frac{1}{4} \frac{Q_i^\ell+Q_j^{\ell-1}}{Q_i^\ell-Q_j^{\ell-1}}, & (\alpha,\beta) = (0,\ell-1) \\
            \frac{1}{4} \frac{Q_i^{\ell-1}+Q_j^0}{Q_i^{\ell-1}-Q_j^0} + \frac{1}{4} \frac{Q_i^{\ell-1}+Q_j^\ell}{Q_i^{\ell-1}-Q_j^\ell}, & (\alpha,\beta) = (\ell-1,0) \\
            \frac{1}{2} \frac{Q_i^\alpha+Q_j^\beta}{Q_i^\alpha-Q_j^\beta}, & \text{otherwise},
        \end{cases}
    \end{equation*}
    and for $\ell=1$, it is
    \begin{equation*}
        \{ u_i^{0\pm},u_j^{0\pm} \} = \pm(1-\delta_{ij})u_i^{0\pm} u_j^{0\pm} \bigg( \frac{Q_i^0+Q_j^0}{Q_i^0-Q_j^0} - \frac{1}{2} \frac{Q_i^1+Q_j^0}{Q_i^1-Q_j^0} - \frac{1}{2} \frac{Q_i^0+Q_j^1}{Q_i^0-Q_j^1} \bigg).
    \end{equation*}
    Since $\frac{Q_i^\ell+Q_j^{\ell-1}}{Q_i^\ell-Q_j^{\ell-1}} \equiv \frac{Q_i^0+Q_j^{\ell-1}}{Q_i^0-Q_j^{\ell-1}} \mod (t-1) \mathfrak{A}_{N,\ell}^K$, the result follows.
\end{proof}

\subsection{Quantum toroidal algebra of $\mathfrak{gl}_\ell$}

In the case of the $K$-theoretic Coulomb branch, the affine Yangian is replaced by the quantum toroidal algebra. To see this, we introduce the generating series
\begin{equation}
    e^\alpha(z) \coloneq \sum_{i=1}^N \frac{u_i^{\alpha+}}{1-Q_i^\alpha/z} \in \mathfrak{A}_{N,\ell}^K\llbracket z^{-1} \rrbracket, \qquad
    f^\alpha(z) \coloneq \sum_{i=1}^N \frac{u_i^{\alpha-}}{1-Q_i^\alpha/z} \in \mathfrak{A}_{N,\ell}^K\llbracket z^{-1} \rrbracket,
\end{equation}
as well as
\begin{equation}\label{eq:modKthGaugePoly}
    h^\alpha(z) \coloneq \frac{\tilde \chi^{\alpha+1}(z) \tilde \chi^{\alpha-1}(z)}{\tilde \chi^\alpha(z)^2} \prod_{i=1}^N \frac{Q_i^\alpha}{(Q_i^{\alpha+1} Q_i^{\alpha-1})^{1/2}}, \qquad \tilde \chi^\alpha(z) \coloneq \prod_{i=1}^N (1-Q_i^\alpha/z).
\end{equation}
We then expand them according to
\begin{equation}
    e^\alpha(z) = \sum_{r \geq 0} e_r^\alpha z^{-r}, \qquad f^\alpha(z) = \sum_{r \geq 0} f_r^\alpha z^{-r}, \qquad h^\alpha(z) = \sum_{r \geq 0} k_r^{\alpha+} z^{-r} = \sum_{r \geq 0} k_r^{\alpha-} z^r.
\end{equation}

\begin{prop}
    The generating series $e^\alpha(z),f^\alpha(z),\tilde \chi^\alpha(z),h^\alpha(z)$ define a representation of the classical limit of the $N$-truncated positive half of the quantum toroidal algebra of $\mathfrak{gl}_\ell$ in the sense that
    \begin{align}
        \{ k_r^\alpha, k_s^\beta \} &= 0,
        \\[8pt]
        \{ e_r^\alpha,f_s^\beta \}
        &= \delta^{\alpha\beta} (k_{r+s}^{\alpha+} - k_{r+s}^{\alpha-}),
        \\[2pt]
        \{ \tilde \chi^\alpha(z),e^\beta(w) \}
        &= \delta^{\alpha\beta} \tilde \chi^\alpha(z) \frac{e^\beta(z)-e^\beta(w)}{z/w-1},
        \\[2pt]
        \{ \tilde \chi^\alpha(z),f^\beta(w) \}
        &= -\delta^{\alpha\beta} \tilde \chi^\alpha(z) \frac{f^\beta(z)-f^\beta(w)}{z/w-1}.
    \end{align}
\end{prop}

\begin{remark}
    This is the classical limit of the GKLO representation of the quantum toroidal algebra of $\mathfrak{gl}_\ell$ discussed in \cite{tsymbaliuk:2023}.
\end{remark}

\begin{corollary}
    The zero modes
    \begin{equation}
        E^\alpha \coloneq e_0^\alpha = \sum_{i=1}^N u_i^{\alpha+}, \quad F^\alpha \coloneq f_0^\alpha = \sum_{i=1}^N u_i^{\alpha-}, \quad K^\alpha \coloneq k_0^{\alpha+} = (k_0^{\alpha-})^{-1} = \prod_{i=1}^N \frac{Q_i^\alpha}{(Q_i^{\alpha+1} Q_i^{\alpha-1})^{1/2}}
    \end{equation}
    define a representation of the classical limit of the quantum affine algebra in the sense that
    \begin{align}
        \{ K^\alpha,K^\beta \} &= 0,
        \\[4pt]
        \{ K^\alpha,E^\beta \} &= \tfrac{1}{2} \kappa^{\alpha\beta} K^\alpha E^\beta,
        \\[4pt]
        \{ K^\alpha, F^\beta\} &= -\tfrac{1}{2} \kappa^{\alpha\beta} K^\alpha F^\beta,
        \\[4pt]
        \{ E^\alpha,F^\beta \} &= \delta^{\alpha\beta} (K^\alpha - (K^\alpha)^{-1}).
    \end{align}
\end{corollary}

\subsection{$L$-operator algebra}

As in the rational case, we can write down an algebra of $L$-operators:

\begin{definition}
    Introduce the one-site $L$-operators
    \begin{equation}
        L_{ij}^\alpha \coloneq \frac{u_j^{\alpha+1,+}}{1-Q_j^{\alpha+1}/Q_i^\alpha} \in \mathfrak{A}_{N,\ell}^K
    \end{equation}
    as well as the total $L$-operator
    \begin{equation}\label{eq:totLHyp}
        L \coloneq L^0 \cdots L^{\ell-1}.
    \end{equation}
\end{definition}

\begin{prop}
The $L$-operators satisfy the Poisson bracket
\begin{equation}\label{eq:Lops}
    \{ L_1^\alpha, L_2^\beta \} = \delta^{\alpha\beta} r^\alpha L_1^\alpha L_2^\beta - \delta^{\alpha\beta} L_1^\alpha L_2^\beta \underline r^{\alpha+1} + \delta^{\alpha+1,\beta} L_1^\alpha \bar r_{21}^\beta L_2^\beta - \delta^{\alpha,\beta+1} L_2^\beta \bar r^\alpha L_1^\alpha,
\end{equation}
where we have used the matrices from \cite{arutyunov:2019b}, except that $\bar r^\alpha$ is shifted by $-\tfrac{1}{2}$:
\begin{align}\label{eq:hypMat}
    r^\alpha &\coloneq \sum_{i \neq j} \Big( \frac{1}{Q_i^\alpha/Q_j^\alpha-1} e_{ii} - \frac{1}{1-Q_j^\alpha/Q_i^\alpha} e_{ij} \Big) \otimes (e_{jj} - e_{ji}),
    \\
    \bar r^\alpha &\coloneq \sum_{i \neq j} \frac{1}{1-Q_j^\alpha/Q_i^\alpha} (e_{ii}-e_{ij}) \otimes e_{jj} - \frac{1}{2},
    \\[4pt]
    \underline r^\alpha &\coloneq \sum_{i \neq j} \frac{1}{1-Q_j^\alpha/Q_i^\alpha} (e_{ij} \otimes e_{ji} - e_{ii} \otimes e_{jj}).
\end{align}
\end{prop}

\begin{remark}
    The $r$-matrices $r^\alpha,\underline r^\alpha$ are not anti-symmetric, but satisfy
    \begin{equation}
        r^\alpha + r_{21}^\alpha = C - 1, \qquad \underline r^\alpha + \underline r_{21}^\alpha = C - 1.
    \end{equation}
    with $C = \sum_{ij} e_{ij} \otimes e_{ji}$.
\end{remark}

\begin{corollary}
    The total $L$-operator satisfies
    \begin{equation}
        \{ L_1, L_2 \} = r^0 L_1 L_2 - L_1 L_2 \underline r^0 + L_1 \bar r_{21}^0 L_2 - L_2 \bar r^0 L_1,
    \end{equation}
    which reproduces the Poisson bracket of the Lax matrix from \cite{arutyunov:2019b}.
\end{corollary}

\begin{proof}
    This follows from the Poisson bracket of the one-site $L$-operators using the identity
    \begin{equation*}
        r^\alpha + \bar r_{21}^\alpha - \bar r^\alpha - \underline r^\alpha = 0.
        \vspace{-1.5\baselineskip}
    \end{equation*}
\end{proof}

\begin{corollary}
    The Hamiltonians $H[n] \coloneq \operatorname{Tr} L^n$ are mutually Poisson commuting.
\end{corollary}

\begin{prop}
    The Hamiltonians $H[n]$ are central in the classical limit of the quantum affine algebra:
    \begin{equation}
        \{ H[n], E^\alpha \} = 0, \qquad \{ H[n], F^\alpha \} = 0, \qquad \{ H[n], K^\alpha \} = 0.
    \end{equation}
\end{prop}

\begin{proof}
    Let us consider the bracket $\{ H[n],E^\alpha \}$ as an example. From the $L$-operator algebra, we derive the bracket
    \begin{equation*}
        \begin{aligned}
            \{ L_1^\alpha,u_2^{\beta+1,+} \}
            ={}& {-\delta^{\alpha\beta}} L_1^\alpha u_2^{\beta+1,+} \underline r^{\alpha+1} - \delta^{\alpha,\beta+1} u_2^{\beta+1,+} (\bar r^\alpha+\tfrac{1}{2}) L_1^\alpha + \delta^{\alpha\beta} Z L_1^\alpha \tilde L_2^\beta - \delta^{\alpha+1,\beta} L_1^\alpha Z \tilde L_2^\beta.
        \end{aligned}
    \end{equation*}
    with $Z = \sum_{i=1}^N e_{ii} \otimes e_i^t$. Then
    \begin{equation*}
        \begin{aligned}
            \{ H[n],E^\alpha \}
            ={}& \sum_{\mu=0}^{n\ell-1} \operatorname{Tr}_1 L_1^0 \cdots L_1^{\mu-1} \{ L_1^\mu, u_2^{\alpha,+} \} L_1^{\mu+1} \cdots L_1^{n\ell-1} e_2 \\
            ={}& -\sum_{\mu=1}^{n\ell-1} \delta^{\mu\alpha} \operatorname{Tr}_1 L_1^0 \cdots L_1^{\alpha-1} u_2^{\alpha,+} (\underline r^{\alpha} + \bar r^\alpha) L_1^{\alpha} \cdots L_1^{n\ell-1} e_2 - \delta^{\ell\alpha} \operatorname{Tr}_1 L_1^n u_2^{\alpha,+} (\underline r^{0}+\bar r^0) e_2 \\
            &- \sum_{\mu=2}^{n\ell+1} \tfrac{1}{2} \delta^{\mu\alpha} \operatorname{Tr}_1 L_1^n u_2^{\alpha,+} e_2 + \delta^{0,\alpha-1} \operatorname{Tr}_1 (Z \tilde L_2^{\alpha-1} L_1^n-L_1^n Z \tilde L_2^{\alpha-1}) e_2 \\
            ={}& 0,
        \end{aligned}
    \end{equation*}
    where we have used $(\underline r^\alpha + \bar r^\alpha) e_2 = -\tfrac{1}{2} e_2$.
\end{proof}

In analogy with the rational case, this suggests that the central Hamiltonians $H[n]$ together with the rest of the Bethe subalgebra of the classical limit of the quantum affine algebra define a superintegrable system with the $K$-theoretic Coulomb branch as its phase space. In the next section, we identify this superintegrable model with the hyperbolic spin RS model.

\subsection{Equations of motion}

Let us study the equations of motion under the first Hamiltonian $H[1]$. To this end, we let $Q^\alpha \coloneq \operatorname{diag}(Q_1^\alpha,\dots,Q_N^\alpha)$ and $\tilde L^\alpha \coloneq (Q^\alpha)^{-1} L^\alpha Q^{\alpha+1}$. Then we define the spin vectors
\begin{equation}
    a^\alpha \coloneq L^0 \cdots L^{\alpha-2} e, \qquad c^\alpha \coloneq u^{\alpha+} \tilde L^\alpha \cdots \tilde L^{\ell-1}.
\end{equation}
for $\alpha=1,\dots,\ell$. Notice again that $a_i^1 = 1$ for $i=1,\dots,N$.

\begin{lemma}\label{lemma:hypLOp}
    The total $L$-operator can be written as
    \begin{equation}
    L_{ij} = \frac{\sum_{\rho=1}^\ell a_i^\rho c_j^\rho}{1-Q_j^\ell/Q_i^0}.
\end{equation}
\end{lemma}

\begin{proof}
    Indeed,
    \begin{equation*}
        \begin{aligned}
            \sum_{\rho=1}^\ell &a_i^\rho c_j^\rho
            = \sum_{\rho=1}^\ell \sum_{k,l=1}^N (L^0 \cdots L^{\rho-2})_{il} u_k^{\rho+}  Q_j^\ell/Q_k^\rho (L^\rho \cdots L^{\ell-1})_{kj} \\
            ={}& \sum_{\rho=1}^\ell \sum_{k,l=1}^N (L^0 \cdots L^{\rho-2})_{il} L_{lk}^{\rho-1} (Q_j^\ell/Q_k^\rho-Q_j^\ell/Q_l^{\rho-1}) (L^\rho \cdots L^{\ell-1})_{kj} \\
            ={}& \sum_{\rho=1}^\ell \sum_{k=1}^N (L^0 \cdots L^{\rho-1})_{ik} (L^\rho \cdots L^{\ell-1})_{kj} Q_j^\ell/Q_k^\rho - \sum_{\rho=1}^\ell \sum_{l=1}^N (L^0 \cdots L^{\rho-2})_{il} (L^{\rho-1} \cdots L^{\ell-1})_{lj} Q_j^\ell/Q_l^{\rho-1} \\
            ={}& L_{ij} (1 - Q_j^\ell/Q_i^0),
        \end{aligned}
    \end{equation*}
    which yields the result.
\end{proof}

\begin{theorem}
The time evolution under the Hamiltonian $H \coloneq (t-1) H[1]$ reproduces the equations of motion \eqref{eq:eom} with the identification $x_i = \log Q_i^0, \gamma = -{\log t}$ and the hyperbolic potential $V(z) = \tfrac{1}{2} \coth \tfrac{z}{2}-\tfrac{1}{2} \coth \tfrac{z+\gamma}{2}$.
\end{theorem}

\begin{proof}
    We use the $L$-operator algebra to derive the equations of motion
    \begin{align*}
        \dot x_i &= \{ H,Q_i^0 \} / Q_i^0 = -(t-1)L_{ii},
        \\[4pt]
        \ddot x_i &= \{ H, \{ H,Q_i^0 \} / Q_i^0 \} = (t-1)^2 \sum_{i \neq j} L_{ij} L_{ji} \frac{Q_i^0+Q_j^0}{Q_i^0-Q_j^0}, \\
        \dot a_i^\alpha &= \{ H,a_i^\alpha \} = -(t-1)\sum_{i \neq j} (a_i^\alpha-a_j^\alpha) \frac{L_{ij}}{1-Q_i^0/Q_j^0}, \\
        \dot c_i^\alpha &= \{ H,c_i^\alpha \} = (t-1)\sum_{i \neq j} \Big( c_i^\alpha \frac{L_{ij}}{1-Q_i^0/Q_j^0} - c_j^\alpha \frac{L_{ji}}{1-Q_j^0/Q_i^0} \Big).
    \end{align*}
    We then use lemma \ref{lemma:hypLOp} and the identities
    \begin{align*}
        (t-1)\frac{1}{1-Q_i^0/Q_j^0} \frac{1}{1-Q_j^\ell/Q_i^0} &= V(x_i-x_j), \\
        (t-1)^2\frac{Q_i^0+Q_j^0}{Q_i^0-Q_j^0} \frac{1}{1-Q_j^\ell/Q_i^0} \frac{1}{1-Q_i^\ell/Q_j^0} &= V(x_i-x_j)-V(x_j-x_i)
    \end{align*}
    to arrive at the desired equations of motion \eqref{eq:eom}.
\end{proof}

\subsection{Poisson bracket of the spin variables}

\begin{prop}
    The Poisson brackets of the spins can be written as
    \begin{align}
        \{ Q_i^0, a_j^\alpha \} ={}& 0, \qquad \{ Q_i^0, c_j^\alpha \} = \delta_{ij} Q_i^0 c_j^\alpha, \\[4pt]
        \{ a_i^\alpha,a_j^\beta \} ={}& (1-\delta_{ij}) \bigg( \frac{a_i^\alpha (a_j^\beta-a_i^\beta)}{Q_i^0/Q_j^0-1} - \frac{a_j^\alpha (a_j^\beta-a_i^\beta)}{1-Q_j^0/Q_i^0} \bigg) - \delta^{\alpha < \beta} a_j^\alpha a_i^\beta \\
        &+ \tfrac{1}{2} (2\delta^{1=\alpha<\beta} + \delta^{\alpha>\beta>1} + \delta^{1 < \alpha < \beta}) a_i^\alpha a_j^\beta, \nonumber
        \\[4pt]
        \{ a_i^\alpha,c_j^\beta \} ={}& (1-\delta_{ij}) \frac{(a_j^\alpha-a_i^\alpha) c_j^\beta}{1-Q_j^0/Q_i^0} - \delta^{\alpha\beta} \sum_{\rho=1}^{\alpha-1} a_i^\rho c_j^\rho + \delta^{1\beta} \tilde L_{ij} a_i^\alpha - \delta^{\alpha\beta} \tilde L_{ij} \\
        &+ \tfrac{1}{2}(1 - \delta^{1=\alpha < \beta} + \delta^{\alpha > \beta=1}-\delta^{\alpha\beta}) a_i^\alpha c_j^\beta, \nonumber
        \\[4pt]
        \{ c_i^\alpha,c_j^\beta \} ={}& (1-\delta_{ij}) \bigg( \frac{c_j^\alpha c_i^\beta}{Q_i^0/Q_j^0-1} + \frac{c_i^\alpha c_j^\beta}{1-Q_j^0/Q_i^0} \bigg) + \delta^{\alpha < \beta} c_j^\alpha c_i^\beta - \delta^{1\beta} c_i^\alpha \tilde L_{ij} + \delta^{1\alpha} c_j^\beta \tilde L_{ji} \\
        &- \tfrac{1}{2} (2 \delta^{\alpha > \beta=1} + \delta^{\alpha > \beta > 1} + \delta^{1<\alpha<\beta}) c_i^\alpha c_j^\beta, \nonumber
    \end{align}
    where
    \begin{equation}
        \tilde L_{ij} = \frac{\sum_{\rho=1}^\ell a_i^\rho c_j^\rho}{Q_i^0/Q_j^\ell-1}.
    \end{equation}
\end{prop}

\begin{proof}
    We again consider the bracket $\{ a_1^\alpha,a_2^\beta \}$ as an example and make use of the partial monodromies $L^{\alpha,\beta} \coloneq L^\alpha \cdots L^\beta$. Here $1$ and $2$ again label auxiliary spaces. Then
    \begin{equation*}
        \begin{aligned}
            \{ &a_1^\alpha,a_2^\beta \} = \sum_{\mu=0}^{\alpha-2} \sum_{\nu=0}^{\beta-2} L_1^{0,\mu-1} L_2^{0,\nu-1} \{ L_1^\mu, L_2^\nu \} L_1^{\mu+1,\alpha-2} L_2^{\nu+1,\beta-2} e_1 e_2 \\
            ={}& \sum_{\mu=0}^{\alpha-2} \sum_{\nu=0}^{\beta-2} \delta^{\mu\nu} L_1^{0,\mu-1} L_2^{0,\nu-1} r^\mu L_1^{\mu,\alpha-2} L_2^{\nu,\beta-2} e_1 e_2 - \sum_{\mu=1}^{\alpha-1} \sum_{\nu=1}^{\beta-1} \delta^{\mu\nu} L_1^{0,\mu-1} L_2^{0,\nu-1} \underline r^\mu L_1^{\mu,\alpha-2} L_2^{\nu,\beta-2} e_1 e_2 \\
            &+ \sum_{\mu=1}^{\alpha-1} \sum_{\nu=1}^{\beta-2} \delta^{\mu\nu} L_1^{0,\mu-1} L_2^{0,\nu-1} \bar r_{21}^\mu L_1^{\mu,\alpha-2} L_2^{\nu,\beta-2} e_1 e_2 - \sum_{\mu=1}^{\alpha-2} \sum_{\nu=1}^{\beta-1} \delta^{\mu\nu} L_1^{0,\mu-1} L_2^{0,\nu-1} \bar r^\mu L_1^{\mu,\alpha-2} L_2^{\nu,\beta-2} e_1 e_2 \\
            ={}& (r^0 - \delta_{1<\alpha<\beta} P + \tfrac{1}{2}\delta_{1<\alpha<\beta} + \tfrac{1}{2} \delta_{1<\beta<\alpha}) a_1^\alpha a_2^\beta
        \end{aligned}
    \end{equation*}
    where we used the identities
    \begin{equation*}
        \underline r^\alpha e_1 e_2 = 0, \qquad (\underline r^\alpha + \bar r^\alpha) e_2 = -\tfrac{1}{2} e_2, \qquad (\underline r^\alpha - \bar r_{21}^\alpha - P) e_1 = -\tfrac{1}{2} e_1.
    \end{equation*}
    The other brackets follow similarly, except that we have to use the additional identity
    \begin{equation*}
        L^0 \cdots L^{\alpha-2} \tilde L^{\alpha-1} \cdots \tilde L^{\ell-1} = \tilde L + \sum_{\rho=1}^{\alpha-1} a^\rho c^\rho,
    \end{equation*}
    which can be proven by a similar telescopic argument as in the proof of lemma \ref{lemma:hypLOp}. All in all, we obtain
    \begin{align}
        \{ a_1^\alpha,a_2^\beta \} ={}& (r^0 - \delta_{\alpha<\beta} P) a_1^\alpha a_2^\beta + (\delta_{1=\alpha<\beta} + \tfrac{1}{2} \delta_{\alpha>\beta>1} + \tfrac{1}{2}\delta_{1<\alpha<\beta}) a_1^\alpha a_2^\beta,
        \\[6pt]
        \{ a_1^\alpha,c_2^\beta \} ={}& {-c_2^\beta} \bar r^0 a_1^\alpha - (\tfrac{1}{2} \delta^{\alpha\beta} + \tfrac{1}{2} \delta_{1=\alpha < \beta} - \tfrac{1}{2} \delta_{\alpha > \beta=1}) a_1^\alpha c_2^\beta - \delta^{\alpha\beta} \sum_{\rho=1}^{\alpha-1} a_1^\rho c_2^\rho \\
        &+ \delta^{1\beta} Z a_1^\alpha \tilde L_2 - \delta^{\alpha\beta} Z e_1 \tilde L_2, \nonumber
        \\[6pt]
        \{ c_1^\alpha,c_2^\beta \} ={}& {-c_1^\alpha} c_2^\beta (\underline r^0 - \delta_{\alpha<\beta} P) - (\delta_{\alpha>\beta=1} + \tfrac{1}{2} \delta_{\alpha>\beta>1} + \tfrac{1}{2} \delta_{1<\alpha<\beta}) c_1^\alpha c_2^\beta \\
        &- \delta^{1\beta} c_1^\alpha Z \tilde L_2 + \delta^{1\alpha} c_2^\beta Z_{21} \tilde L_1, \nonumber
    \end{align}
    where $Z = \sum_{i=1}^N e_{ii} \otimes e_i^t$ and $P = \sum_{ij} e_{ij} \otimes e_{ji}$. Writing this in components, we arrive at the claimed brackets.
\end{proof}

\begin{remark}
    We again have the constraint $a_i^1 = 1$, which restricts the ``physical'' spin degrees of freedom of the hyperbolic spin RS model to $a^\alpha,c^\alpha$ for $\alpha > 1$ and notice that they form a quadratic Poisson algebra.
\end{remark}

\begin{corollary}
    \def\a{\alpha}
    \def\b{\beta}
    \def\sgn{\rm sgn}
    The rescaling $\hat a_i^\alpha \coloneq a_i^\alpha/\sum_\rho a_i^\rho, \hat c_i^\alpha \coloneq c_i^\alpha \sum_\rho a_i^\rho, \hat{ \tilde{L}}_{ij} \coloneq \frac{\sum_\rho \hat a_i^\rho \hat c_j^\rho}{Q_i^0/Q_j^\ell-1}$ reproduces the Poisson brackets from \cite{arutyunov:2019,fairon:2026}:
    \begin{align}
            \{ Q_i^0, \hat a_j^\alpha \} ={}& 0, \qquad \{ Q_i^0, \hat c_j^\alpha \} = \delta_{ij} Q_i^0 \hat c_j^\alpha, \nonumber \\[4pt]
            \{\hat{a}_i^{\a},\hat{a}_j^{\b} \}
            ={}&-(1-\delta_{ij})\frac{1}{2}\frac{Q_i^0+Q_j^0}{Q_i^0-Q_j^0}(\hat{a}_i^{\a}-\hat{a}_j^{\a})(\hat{a}_i^{\b}-\hat{a}_j^{\b})-\frac{\sgn(\beta-\a)}{2}\hat{a}_j^{\a}\hat{a}_i^{\b} \nonumber \\
            &-\sum_{\rho=1}^{\ell}\frac{\sgn(\a-\rho)}{2}\hat{a}_j^{\a}\hat{a}_j^{\b}\hat{a}_i^{\rho}+\sum_{\rho=1}^{\ell}\frac{\sgn(\b-\rho)}{2}\hat{a}_i^{\a}\hat{a}_i^{\b}\hat{a}_j^{\rho}
            -\sum_{\mu,\nu=1}^{\ell}\frac{\sgn(\mu-\nu)}{2}\hat{a}_i^{\mu}\hat{a}_j^{\nu} \hat{a}_i^{\a}\hat{a}_j^{\b}\, , \nonumber \\
            \{\hat{a}_i^{\a},\hat{c}_j^{\b} \}={}&\hat{a}_i^{\a} \hat{ \tilde{L}}_{ij}-(1-\delta_{ij})\frac{1}{2}\frac{Q_i^0+Q_j^0}{Q_i^0-Q_j^0}(\hat{a}_i^{\a}-\hat{a}_j^{\a})\hat{c}_j^{\b}+\hat{a}_i^{\a}\sum_{\rho=1}^{\beta-1}\hat{a}_i^{\rho}\hat{c}_j^{\rho}+\frac{1}{2}\hat{a}_i^{\a}\hat{a}_i^{\b}\hat{c}_j^{\b} \nonumber \\
            &-\delta_{\a\b}\bigg(\hat{ \tilde{L}}_{ij}+\frac{1}{2}\hat{a}_i^{\a}\hat{c}_j^{\b}+\sum_{\rho=1}^{\beta-1}\hat{a}_i^{\rho}\hat{c}_j^{\rho}\bigg)
            +\sum_{\rho=1}^{\ell}\frac{\sgn(\a-\rho)}{2}\hat{c}_{j}^{\b}\hat{a}_j^{\a}\hat{a}_i^{\rho}
            +\sum_{\mu,\nu=1}^{\ell}\frac{\sgn(\mu-\nu)}{2}\hat{a}_i^{\mu}\hat{a}_j^{\nu}\hat{a}_i^{\a}\hat{c}_j^{\b}\, , \displaybreak \nonumber \\
            \{\hat{c}_i^{\a},\hat{c}_j^{\b} \}={}&\frac{1}{2}(1-\delta_{ij})\frac{Q_i^0+Q_j^0}{Q_i^0-Q_j^0}(\hat{c}_i^{\a}\hat{c}_j^{\b}+\hat{c}_j^{\a}\hat{c}_i^{\b})-\hat{c}_i^{\a} \hat{ \tilde{L}}_{ij}+\hat{c}_j^{\b} \hat{ \tilde{L}}_{ji}+\frac{\sgn(\b-\a)}{2}\hat{c}_j^{\a}\hat{c}_i^{\b} \nonumber \\
            &-\hat{c}_i^{\a}\sum_{\rho=1}^{\b-1}\hat{a}_i^{\rho}\hat{c}_j^{\rho}+\hat{c}_j^{\b}\sum_{\rho=1}^{\a-1}\hat{a}_j^{\rho}\hat{c}_i^{\rho}
            -\sum_{\mu,\nu=1}^{\ell}\frac{\sgn(\mu-\nu)}{2}\hat{a}_i^{\mu}\hat{a}_j^{\nu} \hat{c}_i^{\a}\hat{c}_j^{\b}+\frac{1}{2}\hat{a}_j^{\a}\hat{c}_i^{\a}\hat{c}_j^{\b}-\frac{1}{2}\hat{a}_i^{\b}\hat{c}_{i}^{\a}\hat{c}_j^{\b}\, .
    \end{align}
\end{corollary}

\begin{remark}
    The fact that our Poisson brackets coincide with the ones in \cite{arutyunov:2019,fairon:2026} can be viewed as an instance of mirror symmetry. The principle of mirror symmetry posits a correspondence between the multiplicative quiver variety of the Jordan quiver considered in \cite{fairon:2026} and the $K$-theoretic Coulomb branch of the necklace quiver considered here.
\end{remark}

\subsection{Supplementary Poisson brackets}

In this section, we give supplementary Poisson brackets of certain variables introduced in \cite{arutyunov:1998}:
\begin{equation}
    S_i^{\alpha\beta} \coloneq c_i^\alpha a_i^\beta, \qquad f_{ij} \coloneq \sum_{\rho=1}^\ell a_i^\rho c_j^\rho.
\end{equation}
The variables $S_i^{\alpha\beta}$ were dubbed \emph{quantum current} in \cite{arutyunov:2025} due to the form of their quantum commutation relations, which are similar to the commutation relations of the quantum current defined in \cite{reshetikhin:1990}. The variables $f_{ij}$ are the \emph{collective spin variables}. They Poisson-commute with $Q_i^0$ according to
\begin{eqnarray}
    \{ Q_i^0,S_j^{\alpha\beta} \} = \delta_{ij} Q_i^0 S_j^{\alpha\beta}, \qquad \{ Q_i^0,f_{jk} \} = \delta_{ik} Q_i^0 f_{jk}.
\end{eqnarray}
We now present the Poisson brackets of the quantum current and collective spin variables of the hyperbolic spin RS model.

\begin{prop}
    \def\a{\alpha}
    \def\b{\beta}
    \def\sgn{\rm sgn}
    The hyperbolic quantum current satisfies the Poisson bracket
    \begin{equation}
        \begin{aligned}
            \{&S_i^{\a\b},S_j^{\mu\nu}\}=(1-\delta_{ij})\frac{1}{2}\frac{Q^0_i+Q^0_j}{Q^0_i-Q^0_j}(S_i^{\mu\beta}S_j^{\a\nu}+S_i^{\a\nu}S_j^{\mu\beta})
            +\frac{\sgn(\mu-\a)}{2}S_i^{\mu\b}S_j^{\a\nu}-\frac{\sgn(\nu-\b)}{2}S_i^{\a\nu}S_j^{\mu\b}
            \\
            &
            -\delta^{\b\mu}\bigg(\tfrac{1}{2}S_i^{\a\b}S_j^{\mu\nu}+\sum_{\rho=1}^{\b-1}S_i^{\a\rho}S_j^{\rho\nu}+\frac{\sum_{\rho=1}^{\ell}S_i^{\a\rho}S_j^{\rho\nu}}{Q_i^0/Q_j^{\ell}-1} \bigg) 
            +\delta^{\a\nu}\bigg(\tfrac{1}{2}S_i^{\a\b}S_j^{\mu\nu}+\sum_{\rho=1}^{\a-1}S_j^{\mu\rho}S_i^{\rho\beta}+\frac{\sum_{\rho=1}^{\ell}S_j^{\mu\rho}S_i^{\rho\beta}}{Q_j^0/Q_i^{\ell}-1} \bigg).
        \end{aligned}
    \end{equation}
\end{prop}

\begin{prop}
    \def\a{\alpha}
    \def\b{\beta}
    The collective spin variables satisfy the Poisson bracket
    \begin{eqnarray}
        \begin{aligned}
            \{ f_{ij}, f_{kl} \}
            &= (1-\delta_{ik})\bigg(\frac{(f_{kl}-f_{il})f_{ij}}{Q_i^0/Q_k^0-1}-\frac{(f_{kl}-f_{il})f_{kj}}{1-Q_k^0/Q_i^0}\bigg) + (1-\delta_{jl})\bigg(\frac{f_{il} f_{kj}}{Q_j^0/Q_l^0-1}+\frac{f_{ij}f_{kl}}{1-Q_l^0/Q_j^0}\bigg) \\
            &- (1-\delta_{kj})\frac{f_{ij}(f_{jl}-f_{kl})}{1-Q_j^0/Q_k^0}+(1-\delta_{il})\frac{f_{kl}(f_{lj}-f_{ij})}{1-Q_l^0/Q_i^0} \\
            &+ \frac{f_{ij} f_{il}}{Q_i^0/Q_l^\ell-1} -\frac{f_{ij} f_{jl}}{Q_j^0/Q_l^\ell-1} + \frac{f_{kl} f_{lj}}{Q_l^0/Q_j^\ell-1} - \frac{f_{kl} f_{kj}}{Q_k^0/Q_j^\ell-1} + \frac{f_{il} f_{kj}}{Q_k^0/Q_j^\ell-1} - \frac{f_{kj} f_{il}}{Q_i^0/Q_l^\ell-1}.
        \end{aligned}
    \end{eqnarray}
\end{prop}

\begin{remark}
    It can be checked that the rational degeneration of these Poisson brackets yields the Poisson bracket found in \cite{arutyunov:1998} and that the form of the Poisson brackets of the collective spin variables matches the conjecture put forth in \emph{loc.\ cit.}
\end{remark}

\section{Conclusion}\label{section:conclusion}

In this paper, we have shown that the (abelianized) homological and $K$-theoretic Coulomb branch Poisson algebras \cite{bullimore:2015,braverman:2018} of the necklace quiver reproduce the equations of motion of the rational and hyperbolic spin RS models \cite{krichever:1995}, respectively. In both cases the same pattern appears: (i) The Coulomb branch monopole operators admit a GKLO realization \cite{gerasimov:2005,finkelberg:2019,tsymbaliuk:2023}, which is a realization in terms of separated canonical variables. (ii) These generators assemble into $L$-operators whose Poisson algebra reproduces the known Poisson algebra of the Lax matrices of the spinless RS models \cite{arutyunov:1996,arutyunov:2019b}. (iii) The central generators $H[n] = \operatorname{Tr} L^n$ supplies a family of commuting Hamiltonians and, together with the (quantum) loop symmetry, yields superintegrability. This pattern matches the quantization of the rational spin RS model found in \cite{arutyunov:2025}, whose algebra of quantum observables conjecturally coincides with the $N$-truncated affine Yangian of $\mathfrak{gl}_\ell$, which is the natural quantization of the necklace quiver Coulomb branch.

These results warrant the conjecture that the elliptic Coulomb branch Poisson algebra \cite{finkelberg:2020} can reproduce the equations of motion of the elliptic spin RS model. It is natural to suspect that this comes about via an $L$-operator algebra of the type discussed above, but where the structure matrices $r^\alpha,\bar r^\alpha, \underline r^\alpha$ are replaced by their elliptic analogs put forth in \cite{arutyunov:1997b}. This provides a clear road map to an explicit description of the Poisson structure and possible quantization of the elliptic spin RS model.

\appendix

\pagebreak

\section{Notation}\label{section:notation}

\begin{tabular}{r|l}
$N$ & number of particles \\
$\ell$ & number of spin polarizations \\
$x_i$ & particle position \\
$a_i^\alpha$ & $\alpha$th component of the spin vector of the $i$th particle \\
$c_i^\alpha$ & $\alpha$th component of the spin covector of the $i$th particle \\
$f_{ij}$ & $\sum_{\alpha=1}^\ell a_i^\alpha c_j^\alpha$ \\
$\gamma$ & coupling constant of the Ruijsenaars--Schneider model \\
$\zeta(z)$ & Weierstrass zeta function \\
$V(z)$ & interaction potential $\zeta(z) - \zeta(z+\gamma)$ and its rational and hyperbolic degenerations \\
$L^\alpha,L^{\alpha\pm}$ & $L$-operators associated to the $\alpha$th gauge node of the necklace quiver \\
$L$ & total $L$-operator of the homological/$K$-theoretic Coulomb branch, \eqref{eq:totLHyp}, \eqref{eq:totLRat} \\
$\mathfrak{C}_{N,\ell}$ & homological Coulomb branch Poisson algebra, Definition \ref{def:cohCoulombAlg} \\
$\mathfrak{C}_{N,\ell}^K$ & $K$-theoretic Coulomb branch Poisson algebra, Definition \ref{def:KthCoulombAlg} \\
$\mathfrak{A}_{N,\ell}$ & rational difference operator algebra, Definition \ref{def:ratDiffOps} \\
$\mathfrak{A}_{N,\ell}^K$ & hyperbolic difference operator algebra, Definition \ref{def:hypDiffOps} \\
$q_i^\alpha$ & scalar vacuum expectation values of the homological Coulomb branch \\
$Q_i^\alpha$ & scalar vacuum expectation values of the $K$-theoretic Coulomb branch \\
$u_i^{\alpha\pm}$ & fundamental monopole operators of the homological/$K$-theoretic Coulomb branch \\
$\chi^\alpha(z)$ & gauge polynomials \eqref{eq:cohGaugePoly} \eqref{eq:KthGaugePoly} of the homological or $K$-theoretic Coulomb branch \\
$\tilde \chi^\alpha(z)$ & modified gauge polynomials \eqref{eq:modKthGaugePoly} of the $K$-theoretic Coulomb branch \\
$P_i^\alpha$ & difference operator in $\mathfrak{A}_{N,\ell}$ and $\mathfrak{A}_{N,\ell}^K$ \\
$e^\alpha(z),f^\alpha(z)$ & raising/lowering generators of the affine Yangian/quantum toroidal algebra inside $\mathfrak{A}_{N,\ell}$/$\mathfrak{A}_{N,\ell}^K$ \\
$h^\alpha(z)$ & Cartan generators of the affine Yangian/quantum toroidal algebra inside $\mathfrak{A}_{N,\ell}$/$\mathfrak{A}_{N,\ell}^K$ \\
$e_r^\alpha,f_r^\alpha$ & modes of the raising/lowering generators of the quantum toroidal algebra inside $\mathfrak{A}_{N,\ell}$/$\mathfrak{A}_{N,\ell}^K$ \\
$k_r^{\alpha\pm}$ & modes of the Cartan generators of the quantum toroidal algebra inside $\mathfrak{A}_{N,\ell}$/$\mathfrak{A}_{N,\ell}^K$ \\
$E^\alpha,F^\alpha$ & raising/lowering generators of the affine/quantum affine $\mathfrak{sl}_\ell$ inside $\mathfrak{A}_{N,\ell}$/$\mathfrak{A}_{N,\ell}^K$ \\
$H^\alpha,K^\alpha$ & Cartan generators of the affine/quantum affine $\mathfrak{sl}_\ell$ inside $\mathfrak{A}_{N,\ell}$/$\mathfrak{A}_{N,\ell}^K$ \\
$r^\alpha,\bar r^\alpha,\underline r^\alpha$ & structure matrices of the $L$-operator algebras \eqref{eq:ratMat}, \eqref{eq:hypMat} \\
$V_{\alpha,\beta}^\pm, J^{\alpha\beta}[0]$ & generators of the loop algebra $L(\mathfrak{gl}_\ell)$ in $\mathfrak{A}_{N,\ell}$ \\
$H[n]$ & Hamiltonians $\operatorname{Tr} L^n$ inside $\mathfrak{A}_{N,\ell}$/$\mathfrak{A}_{N,\ell}^K$ \\
$e$ & vector $(1,\dots,1)^t$ of size $N$ \\
$e_i$ & $i$th basis vector in $\C^N$ \\
$e_{ij}$ & $N \times N$ matrix unit \\
$Z$ & tensor $\sum_{i=1}^N e_{ii} \otimes e_i^t$ \\
$\hat a_i^\alpha$ & rescaled spin vector components $a_i^\alpha / \sum_\rho a_i^\rho$ \\
$\hat c_i^\alpha$ & rescaled spin covector components $c_i^\alpha \sum_\rho a_i^\rho$ \\
$\delta^{\mathcal{P}}$ & one if $\mathcal{P}$ is true and otherwise zero \\
$\kappa^{\alpha\beta}$ & Cartan matrix of the necklace quiver
\end{tabular}

\pagebreak

\vspace{0.3\baselineskip}

\noindent\textbf{Acknowledgments.} Our deepest gratitude goes to J. Teschner for invaluable discussions and guiding LH through the literature on Coulomb branches and to J. Lamers and M. Vasilev for valuable remarks.

\vspace{0.3\baselineskip}

\noindent\textbf{Funding.} GA acknowledges support by the DFG under Germany's Excellence Strategy -- EXC 2121 ``Quantum Universe'' -- 390833306. GA and LH acknowledge support by the DFG -- SFB 1624 -- ``Higher structures, moduli spaces and integrability'' -- 506632645.

\vspace{0.3\baselineskip}

\noindent\textbf{Competing interests.} The authors have no relevant financial or non-financial interests to disclose.

\vspace{0.3\baselineskip}

\noindent\textbf{Data availability.} The authors declare that the data supporting the findings of this study are available within the paper.

\bibliographystyle{alphaurl}
\bibliography{bibliography}
\end{document}